\documentclass[runningheads]{llncs}
\usepackage{amsmath, amssymb}
 
\usepackage{amsthm}
\usepackage{thmtools, thm-restate}
\usepackage{booktabs, array, enumitem, graphicx, colortbl}
\usepackage{subcaption}
\usepackage{hyperref}
\usepackage{tikz}
\usetikzlibrary{positioning,calc,backgrounds,arrows.meta,matrix}
\usepackage{xcolor}

\definecolor{o1Fill}{HTML}{DDF1E7}\definecolor{o1Edge}{HTML}{2E8B6B}
\definecolor{o2Fill}{HTML}{FBEFD4}\definecolor{o2Edge}{HTML}{B8860B}
\definecolor{o3Fill}{HTML}{FBE6E0}\definecolor{o3Edge}{HTML}{C2603F}
\definecolor{lbl}{HTML}{6B7280}

\definecolor{parteal}{HTML}{1D9E75}\definecolor{partealbg}{HTML}{E1F5EE}\definecolor{partealtx}{HTML}{0F6E56}
\definecolor{parcoral}{HTML}{D85A30}\definecolor{parcoralbg}{HTML}{FAECE7}\definecolor{parcoraltx}{HTML}{712B13}
\definecolor{pargray}{HTML}{8A8880}\definecolor{pargraybg}{HTML}{F1EFE8}\definecolor{pargraytx}{HTML}{444441}

\spnewtheorem{assumption}{Assumption}{\bfseries}{\itshape}

\title{LOPMM: A Low-Order Parlay Market Maker}
\author{Niusha Moshrefi\inst{1} \and
        Ranvir Rana\inst{2} \and
        Pramod Viswanath\inst{1}}
\authorrunning{N. Moshrefi et al.}
\institute{Princeton University, Princeton, NJ, USA\\
           \email{\{niusha,pramodv\}@princeton.edu} \and
           Kaleidoscope Blockchain, USA\\
           \email{ranvirranaiitb@gmail.com}}

\begin{document}
\maketitle

\begin{abstract}
Parlays - joint contracts on the simultaneous resolution of several events - are among the most heavily traded products in betting markets, but prediction markets have struggled to offer them natively. The family of parlays over $M$ binary
events is exponentially large, and the mechanisms that could price it either
cannot be computed at that scale or need to pay for the same information multiple times. We start from an empirical observation that
reframes the obstacle. On Kalshi order flow, most of the trades are placed on low-order parlays. For example, on Kalshi 7 days of traded parlay data, $87.03\%$ of the
trades have at most 13 legs and $95.13\%$ have at most 20, with trade counts
bounded by geometrical decay in increasing leg numbers. Given this concentration, and the structural correlation between parlays that
share legs, current designs waste most of either their computation or their
subsidy.

We propose  Low-Order Parlay Market Maker (LOPMM), which keeps
a separate book for each combination of legs to stay tractable, then shares
information from lower-order parlays into the higher-order ones containing
them to reduce the market maker's loss. We formalize the observed
concentration as a low-order flow condition. Empirically, LOPMM substantially reduces the cost of supporting parlay trading. On synthetic low-order flow, its losses remain close to that of running the base markets alone. On a replay of 83 Kalshi NBA games, independent parlay books lose \$174,050 in aggregate while LOPMM earns \$31 under the same liquidity parameter. Its computational cost depends on the order (m) of the requested parlay rather than the number (M) of base events, requiring $\Theta(4^m)$ work per quote and $\Theta(5^m)$ per full trade. We complement these empirical results with theoretical guarantees under informed honest-reporting flow. LOPMMs worst-case loss is bounded by $O(M^2)$ under low-order flow, while its expected loss is bounded by $O(M)$ 
under local dependence. These results characterize when information sharing can make low-order parlay markets substantially cheaper to operate than independent LMSR books.

\end{abstract}

\section{Introduction}
\label{sec:intro}
 
Parlays - joint contracts whose payoff is contingent on the simultaneous
resolution of several underlying events - are among the most heavily traded
products in modern betting markets. In traditional sportsbooks they account
for a large share of activity: roughly $27\%$ of all money wagered across
Illinois, New Jersey, and Colorado through October
2024~\cite{wsj2025parlay}. The same appetite has migrated to prediction
markets. After Kalshi introduced bundled multi-event positions in late
September 2025, parlay volume rose from under $3\%$ to roughly $22\%$ of total
exchange volume in about six months~\cite{sportico2026kalshi}.
 
Prediction markets, however, do not price parlays natively. Because federal
regulation bars these exchanges from taking a side of any trade, both Kalshi
and Polymarket clear parlays through a request-for-quote model in which
third-party market makers supply the opposing
position~\cite{sportico2026kalshi,covers2026combos}; neither operates an
information market over the combinatorial family of parlays. The difficulty
is structural: any specific combination of legs attracts few traders, so each
parlay is an inherently thin market~\cite{ballislife2026parlay}.
 
\paragraph{An empirical observation.}
Thin as each parlay is, the \emph{family} of parlays is not traded uniformly.
Using three datasets collected from Kalshi, we find that parlays with more legs
receive proportionally fewer trades (\S\ref{sec:loworder-empirical}), and the
literature reports the same effect in traditional sportsbooks. The largest
dataset is a 7-day snapshot of trades across all Kalshi markets: $11.4$ million
parlay trades on $5.76$ million parlay markets. Trade counts decay
geometrically in the number of legs, with ratio $0.86$, and $99\%$ of trades
have at most $43$ legs. In short, the information that a parlay market must
absorb decays rapidly in the number of legs, while the number of possible
parlays grows exponentially in the number $M$ of base events.
 
\paragraph{Where the existing designs waste their effort.}
Thin markets are exactly the regime for which the logarithmic market scoring
rule (LMSR) was designed, so that makes it the natural mechanism for pricing
parlays~\cite{hanson2003combinatorial,hanson2007lmsr}. However, two inefficiencies keep it
off the real-world implementation on exchanges. The combinatorial construction runs a single LMSR
over the $2^M$ joint outcomes of the $M$ base events, pricing every parlay
coherently for a subsidy of $b\ln(2^M)=bM\ln2$, linear in $M$ and the same as
the cost of the $M$ base markets alone. Its waste is
\textbf{computational}: pricing is \#P-hard even when traders are restricted
to conjunctions of two events, as is computing the cost of a
transaction~\cite{chen2008complexity}. To quote the two-leg parlay that
$95.7\%$ of the flow asks for, the operator must sum over all $2^M$ joint
states, exponential work spent almost entirely on parlays nobody trades. The
tractable escape is to give each parlay its own book, and then the waste
moves to the \textbf{subsidy}. Parlays on the same set of base markets
are mechanically correlated. Separate books buy the same information
once per book that contains it, so the worst-case subsidy scales with the
number of books rather than with the amount of information. Using $2^M-1$
LMSR markets for the combinatorial space exposes the market maker to an
exponentially growing worst-case loss in the number of base markets.

\begin{figure}[t]
  \centering
  \begin{subfigure}[t]{0.55\linewidth}
    \centering
    \includegraphics[width=\linewidth]{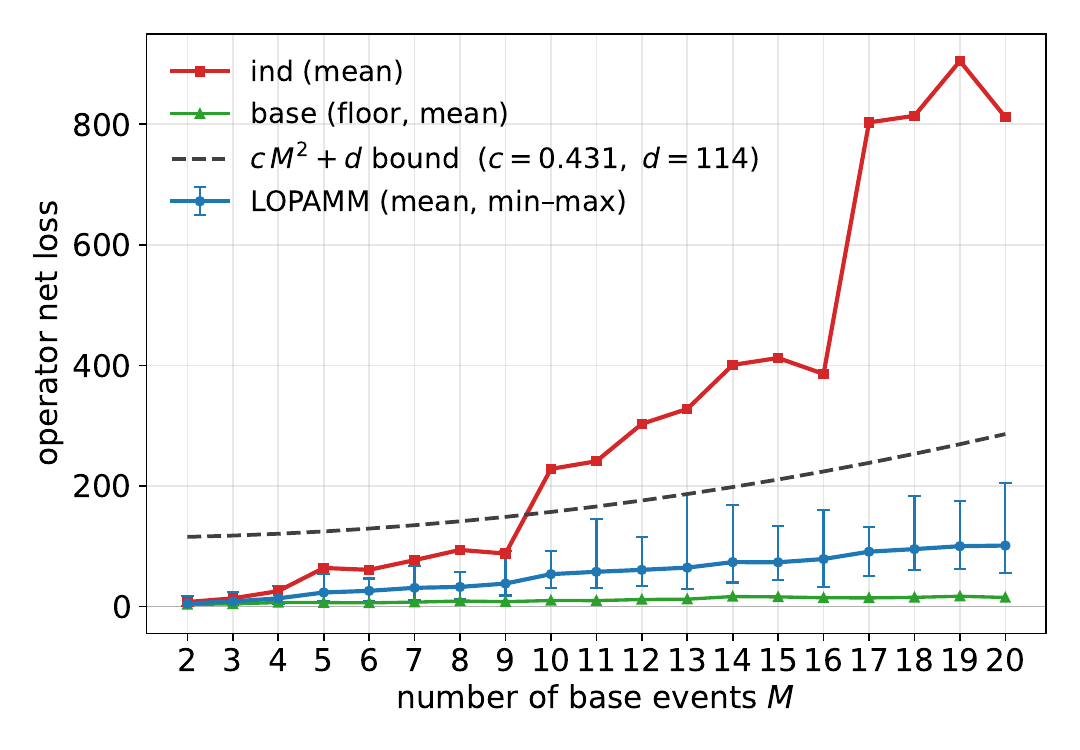}
    \caption{Market maker loss vs.\ $M$ on the same synthetic low-order flow
    ($40$ seeds). LOPMM tracks the cost of the base markets alone and stays
    under its $O(M^2)$ envelope; the independent baseline separates within a
    handful of legs.}
    \label{fig:loss}
  \end{subfigure}
  \hfill
  \begin{subfigure}[t]{0.43\linewidth}
    \centering
    \resizebox{\linewidth}{!}{%
\begin{tikzpicture}[
  >={Latex[length=2.2mm]},
  mnode/.style={rounded corners=3pt, draw, line width=0.6pt,
                minimum width=1.05cm, minimum height=0.72cm,
                font=\small\bfseries, inner sep=2pt},
  traded/.style={mnode, fill=parcoralbg, draw=parcoral, text=parcoraltx},
  updated/.style={mnode, fill=partealbg, draw=parteal, text=partealtx},
  unchanged/.style={mnode, fill=pargraybg, draw=pargray, text=pargraytx},
  lat/.style={draw=pargray, line width=0.6pt},
  prop/.style={draw=parteal, line width=1.1pt, ->},
  olbl/.style={font=\footnotesize, text=pargraytx, anchor=east}
]
  \node[olbl] at (-0.8,0) {order 1};
  \node[olbl] at (-0.8,2) {order 2};
  \node[olbl] at (-0.8,4) {order 3};

  \node[traded]    (A)   at (0,0) {A};
  \node[unchanged] (B)   at (3,0) {B};
  \node[unchanged] (C)   at (6,0) {C};
  \node[updated]   (AB)  at (1,2) {AB};
  \node[updated]   (AC)  at (3,2) {AC};
  \node[unchanged] (BC)  at (5,2) {BC};
  \node[updated]   (ABC) at (3,4) {ABC};

  \draw[lat] (B) -- (AB);
  \draw[lat] (B) -- (BC);
  \draw[lat] (C) -- (AC);
  \draw[lat] (C) -- (BC);
  \draw[lat] (BC) -- (ABC);

  \draw[prop] (A) -- (AB);
  \draw[prop] (A) -- (AC);
  \draw[prop] (AB) -- (ABC);
  \draw[prop] (AC) -- (ABC);

  \draw[draw=parcoral, line width=1.1pt, ->] (0,-1.05) -- (A.south);
  \node[font=\footnotesize, anchor=west, text=parcoraltx] at (0.2,-0.78)
    {incoming trade};

  \matrix[anchor=north, matrix of nodes, column sep=10pt,
          ampersand replacement=\&,
          nodes={anchor=west, font=\footnotesize, inner sep=2pt},
          every node/.style={text=pargraytx}] at (3,-1.7) {
    |[fill=parcoralbg,draw=parcoral,minimum size=8pt,label=right:{traded market}]| {} \&
    |[fill=partealbg,draw=parteal,minimum size=8pt,label=right:{updated}]| {} \&
    |[fill=pargraybg,draw=pargray,minimum size=8pt,label=right:{unchanged}]| {} \\
  };
\end{tikzpicture}}
    \caption{Hierarchical parameterization ($M=3$). A trade on $A$ propagates
    to every parlay containing it.}
    \label{fig:mechanism}
  \end{subfigure}
\end{figure}
 
\paragraph{Our approach: concentrate.}
The observation and the two costs above point to the design that concentrates liquidity and computation where the information and trading activity is: on low-order parlays. We propose the \emph{Low-Order Parlay Market Maker} (LOPMM). Rather than pricing each leg-set in isolation, LOPMM stores one
parameter block per parlay and lets every book read the blocks of the
lower-order parlays its legs contain, a \emph{hierarchical parameterization}
illustrated in Figure~\ref{fig:mechanism} for three base events. The book on ({A,B}), for example, trades against the sum of the blocks for (A), for (B), and a residual block specific to the pair. A trade on (A) therefore changes the state subsequently read by every supported parlay containing (A), allowing lower-order information to propagate through the hierarchy rather than leaving each book fully isolated.

We first evaluate this design empirically. On synthetic low-order flow, LOPMM’s loss remains close to the cost of running the base markets alone. Replaying the recorded Kalshi trade sequences of 83 games through both designs, independent books lose \$174,050 in aggregate while LOPMM earns \$31 under the same liquidity parameter, with LOPMM losing less in 64 of the 83 games. The computational cost of quoting and executing an order-$m$ parlay depends on $m$, not the total number ($M$) of base events. We then characterize theoretically when these gains can be guaranteed. Under informed honest-reporting flow and low-order flow, LOPMM yields an $O(M^2)$ worst-case loss bound, while under local dependence the expected loss is $O(M)$.
 
\paragraph{Our contributions.}
We make four contributions: (i) \textbf{an empirical characterization of parlay order flow:} using Kalshi trade data, we show that parlay trading is concentrated in low-order markets, with trade counts decaying in leg count, and that this pattern holds game by game and across other market categories; we formalize this observation as a \emph{low-order flow} condition on informed trading (Assumption~\ref{ass:low-order}); (ii) \textbf{a market maker tailored to low-order flow:} we propose LOPMM, which propagates lower-order information through a hierarchical parameterization to the higher-order books containing it, without requiring complex optimization; (iii) \textbf{empirical evidence of low market-maker loss:} on both synthetic flow and replayed Kalshi trades, LOPMM keeps the cost of supporting parlay markets close to that of operating the base markets alone and substantially below independent books; and (iv) \textbf{loss guarantees and efficient computation under structured flow:} under the informed-truthful flow setting analyzed in the paper, we establish an $O(M^2)$ worst-case loss bound under low-order flow and an $O(M)$ expected-loss bound under local dependence, while quoting an order-$m$ parlay costs $\Theta(4^m)$ and executing a full trade costs $\Theta(5^m)$, independent of $M$ (Table~\ref{tab:designs}).

\paragraph{Related work.} 
\label{rel-work}

A single LMSR over the $2^M$ joint outcomes prices every parlay coherently
for a subsidy of $bM\ln 2$~\cite{hanson2003combinatorial,hanson2007lmsr},
but pricing is \#P-hard even for two-event
conjunctions~\cite{chen2008complexity}. Prior work restores tractability by
assuming structure in the distribution. Restricted betting
languages~\cite{chen2008tournaments,chen2007permutations,guo2009hierarchies},
Bayesian-network representations of
prices~\cite{pennock2011bayesian,sun2012probability,laskey2018graphical},
or approximate pricing via sampling~\cite{xia2011montecarlo} and convex
relaxation~\cite{abernethy2013efficient,dudik2012tractable,kroer2016arbitrage}.
Dudík et al.~\cite{dudik2012tractable} place these between two poles; the
intractable combinatorial market and the industry practice of independent
securities that share no information. For parlays specifically,
ParlayMarket~\cite{rana2026parlaymarket} fits one pairwise Ising model over
all $M$ events and prices every parlay by belief propagation, bounding the expected loss to $O(M^2)$.

Rather than assume
structure in the distribution, we exploit the structure we measure in the
\emph{flow}. We observe that trades concentrate at low order, so sharing information upward from
low-order books recovers a polynomial subsidy.

\section{Model and Preliminaries}
\label{sec:model}

\subsection{LMSR, cost, and the loss--divergence identity}
\label{sec:lmsr}

Fix a market with finite outcome set $\Omega$, $|\Omega|=n$, run by a
\emph{logarithmic market scoring rule} (LMSR) with liquidity $b>0$
\cite{hanson2003combinatorial,hanson2007lmsr}. The state is a share vector
$q\in\mathbb{R}^{\Omega}$, with \emph{cost function} and \emph{prices}
\begin{equation}
  C(q)=b\log\!\sum_{\omega\in\Omega}\!e^{q_\omega/b},
  \qquad
  \pi_\omega(q)=\frac{e^{q_\omega/b}}{\sum_{\omega'}e^{q_{\omega'}/b}}
  =\partial_{q_\omega}C(q).
  \label{eq:cost-price}
\end{equation}
A move $q\to q'$ costs $C(q')-C(q)$ and issues shares $q'-q$; hence revenue is
path independent. Any trade sequence from $q_0$ to $q$
nets the operator $C(q)-C(q_0)$.


The book opens at an arbitrary state $q_0$ quoting prices $\pi^0:=\pi(q_0)$. If
outcome $\omega$ realizes, the operator's loss is payout minus revenue,
\(
  L(q_0\!\to q;\omega)
  =\big(q_\omega-q_{0,\omega}\big)-\big(C(q)-C(q_0)\big)
  =b\log\frac{\pi_\omega(q)}{\pi^0_\omega},
  \label{eq:loss}
\)
where the last step uses \eqref{eq:cost-price}. The
operator's loss averaged over outcomes under the closing prices is a
Kullback--Leibler divergence, an instance of the Bregman-divergence form of
cost-function market makers~\cite{chen2012utility,abernethy2013efficient}:
\begin{equation}
  \sum_{\omega}\pi_\omega(q)\,L(q_0\!\to q;\omega)
  \;=\;b\sum_{\omega}\pi_\omega(q)\log\frac{\pi_\omega(q)}{\pi^0_\omega}
  \;=\;b\,\mathrm{KL}\!\big(\pi(q)\,\big\|\,\pi^0\big).
  \label{eq:loss-kl}
\end{equation}

\subsection{Parlays as a combinatorial hierarchy}
\label{sec:hierarchy}

Let $[M]=\{1,\dots,M\}$ index $M$ binary base events.
Consider a nonempty leg-set $S\subseteq[M]$ of \emph{order} $|S|$ with outcome set
$\Omega_S=\{0,1\}^S$ ($|\Omega_S|=2^{|S|}$), a \emph{parlay LMSR} is an LMSR book 
with $2^{|S|}$ outcomes, quoting a bet that the legs in $S$ land
on a pattern $\boldsymbol\omega_S\in\Omega_S$.

The parlay structure is a union of parlay LMSRs, one per leg-set. Covering every
$k$-leg combination of base markets takes $\binom{M}{k}$ such markets, so on top of the $M$ base
($k=1$) markets the full hierarchy comprises $\sum_{k=1}^{M}\binom{M}{k}=2^M-1$
LMSR markets, indexed by the non-empty leg-sets $S\subseteq[M]$. Run independently, these $2^M-1$ markets can quote mutually contradictory prices.

\subsection{Latent distribution}
\label{sec:latent}
We posit a latent ground-truth distribution $P^\star\in\Delta(\{0,1\}^M)$ over the
atoms, capturing both the marginal probabilities of the base events and their
dependence. For a non-empty leg-set $S\subseteq[M]$, its marginal is the
distribution $\tau^\star_S\in\Delta(\Omega_S)$ induced on the outcomes of $S$ by
marginalizing $P^\star$:
\begin{equation}
  \tau^\star_S(\boldsymbol\omega_S)
  \;=\!\!\sum_{z:\,z|_S=\boldsymbol\omega_S}\!\!P^\star(z),
  \qquad \boldsymbol\omega_S\in\Omega_S .
  \label{eq:true-marginal}
\end{equation}

\section{Low-order Flow}
\label{sec:loworder-empirical}
We confirm that higher-leg parlays receive proportionally fewer trades using
three datasets collected from Kalshi. First, the NBA same-game parlay trades of the 2026
season from April to June, second, a 12-hour snapshot of every parlay trade across all Kalshi
categories, third, a 7-day snapshot of the same. 

The literature reports the same effect. On Kalshi cross-game NBA parlays, trade
counts fall monotonically from 2 to 11 legs~\cite{moshrefi2026prices}. In New
Jersey sportsbooks in 2020, 2-leg parlays are the most common and counts fall
with each added leg~\cite{nower2023njsports}. In US sportsbook accounts tracked
by Juice Reel, only $10\%$ of parlays reach 10 legs or
more~\cite{covers2024parlay,whelan2026parlay}. Every source shows the same
pattern. Trade counts peak at two or at most a few legs and decay roughly
geometrically beyond that point.

Next, we define the low-order flow formally and show it holds on the data.

\begin{assumption}[Low-order flow]
\label{ass:low-order}
Let $N_m$ be the number of trades on leg-sets of order $m$, counting
repeated trades on the same leg-set. 
There are constants $\rho<1$, $k$, and $\nu$ such that
\begin{equation}
  \label{eq:low-order-quant}
  N_m\;\le\;\rho^{m-2}\,N_{2} \quad (3\le m\le k),
  \qquad
  N_m = 0 \quad (m>k),
  \qquad
  N_2\le\nu\binom M2.
\end{equation}
\end{assumption}

For each dataset we report the three constants of
Assumption~\ref{ass:low-order}. The effective
order $k$ is the smallest order such that parlays of at most $k$ legs account
for $99\%$ of trades. The decay $\rho$ is the smallest value for which
$N_m\le\rho^{\,m-2}N_2$ holds at every order $3\le m\le k$. 

\paragraph{NBA same-game parlays.}
Our corpus is the NBA same-game parlay market trades on Kalshi from
April to June 2026, a span of roughly nine weeks. Across this period $92$ games
carried traded parlays, with $29{,}781$ trades in $7{,}372$ distinct parlay markets,
over $2{,}907$ base markets.
(per-order counts in Table~\ref{tab:flow_nba_app}). 
Two-leg parlays account for $95.7\%$ of
trades. Its constants are $k=7$, $\rho=0.32$, and
$\nu=6.7\times 10^{-3}$.

\paragraph{All Kalshi parlays.}
The two all-category snapshots cover politics, sports, weather, and every
other Kalshi category, and parlays may combine legs from different events. The
12-hour snapshot has $438{,}015$ parlay trades over $152{,}432$ parlay markets on
$M=5{,}251$ base markets. Its constants are $k=19$, $\rho=0.80$, and
$\nu=9.2\times10^{-3}$ (Table~\ref{tab:flow_agg_app}). The 7-day snapshot
has $11.4$ million parlay trades over $5.76$ million parlay markets on
$M=191{,}806$ traded base markets. Its decay is slower: $k=43$, $\rho=0.86$,
and $\nu=9.9\times10^{-5}$ (Figure~\ref{fig:legs7d},
Table~\ref{tab:flow_7d_app}). Parlays of up to 388 legs occur, but only
$0.1\%$ of trades have more than 90 legs.


\begin{figure}[t]
  \centering
  \includegraphics[width=0.7\linewidth]{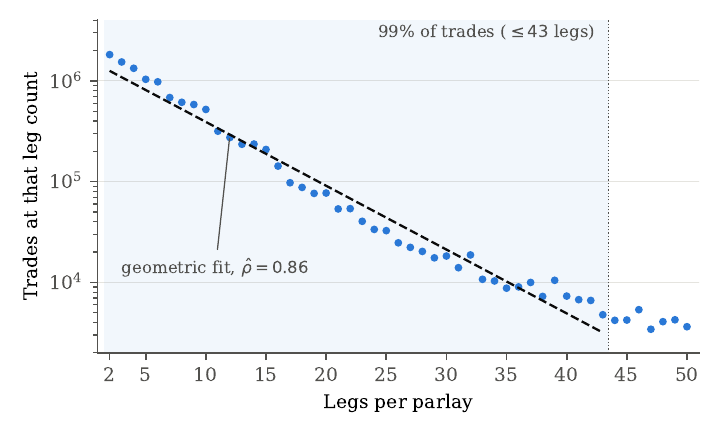}
  \caption{Parlay trades by leg count across all Kalshi,
  September 16--23, 2026 (log scale, so geometric decay appears as a straight line).
  The $0.7\%$ of trades with more than 50 legs are listed in
  Table~\ref{tab:flow_7d_app}.}
  \label{fig:legs7d}
\end{figure}

\paragraph{Time window.} Base markets have finite lives. Once a market
resolves, it can no longer be traded or used as a parlay leg. Over a long
window, most pairs of base markets are never open at the same time, so
extending the window inflates $M$ with markets that cannot be combined into a
parlay. For this reason, we measure all-category flow over at most seven days.
In fact, LOPMM can be defined on top of a much smaller $M$ in practice.
Different sports leagues or independent market categories can have separate
LOPMM structures on top. 

Appendix~\ref{app:supp-data} provides more details on the supporting data.

\section{A Low-Order Parlay Market Maker (LOPMM)}
\label{sec:design}

In this section, we introduce LOPMM, a tractable market design for parlays that works best under 
a low-order trade flow. LOPMM is a union of many LMSR books that share information. 

\subsection{LOPMM's hierarchical parameterization}
\label{sec:method-reparam}

\paragraph{Stored parameters.}
LOPMM stores one parameter per leg combination and its pattern assignment. For each non-empty
$S\subseteq[M]$ and pattern $\boldsymbol\omega_S\in\Omega_S$ we store a parameter $r^{(S)}_{\boldsymbol\omega_S}\in\mathbb{R}$, writing
$q^{(\{k\})}_{\omega}:=r^{(\{k\})}_{\omega}$ for singletons. 
A trade routed to $r^{(S)}_{\boldsymbol\omega_S}$ injects information at
order $|S|$ alone. The number of these parameters is
\begin{equation}
  \label{eq:param-count}
  \bigl|\{r^{(S)}_{\boldsymbol\omega_S}\}\bigr|
  \;=\;\sum_{k=1}^{M}\binom{M}{k}2^k\;=\;3^M-1,
\end{equation}
exactly the count of possible parlay contracts. The parameters are what the operator subsidizes.

\paragraph{Contracts and per-market quotes.}
For each non-empty $S\subseteq[M]$, market $S$ ranges over its $2^{|S|}$ outcomes
$\boldsymbol\omega_S\in\Omega_S$. Its share count for an outcome is the
\emph{aggregate} of its own parameter and every consistent lower-order parameter,
\begin{equation}
  \label{eq:Q-agg}
  Q^{(S)}_{\boldsymbol\omega_S}
  \;=\;\sum_{\emptyset\neq T\subseteq S} r^{(T)}_{\boldsymbol\omega_S|_T},
\end{equation}
where $\boldsymbol\omega_S|_T$ restricts $\boldsymbol\omega_S$ to $T$. Market $S$'s
quote is then read off by the LMSR rule on these aggregates,
\begin{equation}
  \label{eq:cost-design}
  C^{(S)}(\mathbf{Q}^{(S)})
  \;=\;b\log\!\!\sum_{\boldsymbol\omega'_S\in\Omega_S}\!\!
       e^{Q^{(S)}_{\boldsymbol\omega'_S}/b},
  \qquad
  \pi^{(S)}_{\boldsymbol\omega_S}
  \;=\;\frac{e^{Q^{(S)}_{\boldsymbol\omega_S}/b}}
            {\sum_{\boldsymbol\omega'_S}e^{Q^{(S)}_{\boldsymbol\omega'_S}/b}} .
\end{equation}
Prices sum to one within each market, and a base market ($|S|=1$) collapses to the
ordinary binary LMSR. The price a contract carries is this per-market quote
$\pi^{(S)}_{\boldsymbol\omega_S}$. A contract pays out a dollar for outcome resolving to its pattern.

\begin{figure}[t]
\centering
\resizebox{\textwidth}{!}{%
\begin{tikzpicture}[
  font=\sffamily,
  bx/.style={rounded corners=2pt, draw, minimum width=16mm, minimum height=7.5mm,
             inner sep=1pt, anchor=center, line width=0.5pt, font=\scriptsize},
  b1/.style={bx, fill=o1Fill, draw=o1Edge},
  b2/.style={bx, fill=o2Fill, draw=o2Edge},
  b3/.style={bx, fill=o3Fill, draw=o3Edge},
  pls/.style={font=\footnotesize, text=lbl, anchor=center},
  lhs/.style={font=\small, anchor=east},
]

\def\ca{0}    \def\cb{2.0}  \def\cc{4.0}
\def\cd{6.6}  \def\ce{8.6}  \def\cf{10.6}
\def\cg{13.2}
\def\lx{-0.95}     

\def\ya{0}
\node[lhs] at (\lx,\ya) {$Q^{(\{1\})}_{0}\;=$};
\node[b1] at (\ca,\ya) {$q^{(\{1\})}_{0}$};

\def\yb{-1.45}
\node[lhs] at (\lx,\yb) {$Q^{(\{1,2\})}_{01}\;=$};
\node[b1] at (\ca,\yb) {$q^{(\{1\})}_{0}$};
\node[pls] at (1.0,\yb) {$+$};
\node[b1] at (\cb,\yb) {$q^{(\{2\})}_{1}$};
\node[pls] at (4.3,\yb) {$+$};
\node[b2] at (\cd,\yb) {$r^{(\{1,2\})}_{01}$};

\def\yc{-2.90}
\node[lhs] at (\lx,\yc) {$Q^{(\{1,2,3\})}_{010}\;=$};
\node[b1] at (\ca,\yc) {$q^{(\{1\})}_{0}$};   \node[pls] at (1.0,\yc) {$+$};
\node[b1] at (\cb,\yc) {$q^{(\{2\})}_{1}$};   \node[pls] at (3.0,\yc) {$+$};
\node[b1] at (\cc,\yc) {$q^{(\{3\})}_{0}$};   \node[pls] at (5.3,\yc) {$+$};
\node[b2] at (\cd,\yc) {$r^{(\{1,2\})}_{01}$}; \node[pls] at (7.6,\yc) {$+$};
\node[b2] at (\ce,\yc) {$r^{(\{1,3\})}_{00}$}; \node[pls] at (9.6,\yc) {$+$};
\node[b2] at (\cf,\yc) {$r^{(\{2,3\})}_{10}$}; \node[pls] at (11.9,\yc) {$+$};
\node[b3] at (\cg,\yc) {$r^{(\{1,2,3\})}_{010}$};

\begin{scope}[on background layer]
\foreach \cx in {\ca,\cb,\cd} {
  \draw[lbl!35, dashed, line width=0.4pt] (\cx,0.5) -- (\cx,-3.4);
}
\end{scope}

\begin{scope}[yshift=-4.05cm]
  \node[b1, minimum width=5mm, minimum height=4mm, font=\tiny] (lg1) at (0,0) {};
  \node[anchor=west, text=lbl, font=\sffamily\footnotesize] at (lg1.east) {\,Order 1 (singletons)};
  \node[b2, minimum width=5mm, minimum height=4mm, font=\tiny] (lg2) at (4.1,0) {};
  \node[anchor=west, text=lbl, font=\sffamily\footnotesize] at (lg2.east) {\,Order 2 (pairs)};
  \node[b3, minimum width=5mm, minimum height=4mm, font=\tiny] (lg3) at (7.7,0) {};
  \node[anchor=west, text=lbl, font=\sffamily\footnotesize] at (lg3.east) {\,Order 3 (triple)};
\end{scope}

\end{tikzpicture}
}
\caption{The dashed guides mark parameters
reused across markets: a low-order parameter appears in the share count of every
higher-order market whose legs contain it.}
\label{fig:sharing}
\end{figure}

Figure~\ref{fig:sharing} shows the aggregation \eqref{eq:Q-agg}. The share count
the triple market trades against on outcome $010$ stacks its own parameter
$r^{(\{1,2,3\})}_{010}$ on the lower-order parameters consistent with $010$, so a
low-order parameter such as $q^{(\{1\})}_{0}$ or $r^{(\{1,2\})}_{01}$ is reused by
every higher-order market whose legs contain it. 

\subsection{LOPMM's trade mechanics}
\label{sec:trade-mechanics}
A trade begins when an arriving trader hands the operator a target marginal belief
$\tau^{\mathrm{new}}_S\in\Delta(\Omega_S)$ on a single leg-set $S$. The operator realizes this belief on the
trader's behalf as a family of LMSR trades, one per affected market.

\begin{definition}[Full-trade and sub-trades]
\label{def:trade}
A \emph{full-trade} is the operator's response to a trader's target belief. It
consists of a collection of \emph{sub-trades}, at most one for each market $S'\subseteq S$ whose stored
parameter the operator moves, where the sub-trade on market $S'$ updates the single
order-$|S'|$ parameter block $r^{(S')}$ and is priced by that market's cost
\eqref{eq:cost-design}.
\end{definition}

\paragraph{Routing sub-trades.}
The operator issues a sub-trade on every leg-set $S'\subseteq S$, matching each such
book's marginal to the trader's belief, and leaves all other parameters untouched. The
sub-trades are executed bottom-up, in ascending order of $|S'|$, from the base markets
up to $S$ itself. Processing books in this order, every parameter on a strictly smaller
leg-set is already fixed when book $S'$ is reached, so the operator sets the single
block $r^{(S')}$ to drive book $S'$'s quote to the target marginal. Writing the
order-$|S'|$ target in share units,
\begin{equation}
  \label{eq:target}
  g_{S'}(\boldsymbol\omega_{S'})\;:=\;b\log\tau^{\mathrm{new}}_{S'}(\boldsymbol\omega_{S'}),
\end{equation}
with $\tau^{\mathrm{new}}_{S'}$ the marginal of $\tau^{\mathrm{new}}_S$ on $S'$, the
operator sets
\begin{equation}
  \label{eq:route}
  r^{(S')}_{\boldsymbol\omega_{S'}}
  \;\longleftarrow\;
  g_{S'}(\boldsymbol\omega_{S'})\;-\!\!\sum_{\emptyset\neq T\subsetneq S'}\!\!
  r^{(T)}_{\boldsymbol\omega_{S'}|_T} ,
\end{equation}
so \emph{$r^{(S')}$ absorbs exactly the order-$|S'|$ structure of the target marginal} that the
lower-order books already in place do not supply. The recursion is well posed because
its right-hand side depends only on strictly smaller leg-sets, already fixed by the
ascending order.

\paragraph{Buy-only realization.}
The target \eqref{eq:target} and route \eqref{eq:route} determine $r^{(S)}$ up to an
additive constant that shifts shares uniformly on every book $S'\supseteq S$ without
changing prices. Processing books in ascending order of $|S|$, the operator raises
this constant until every outcome's share differential is nonnegative, making each
sub-trade a buy. These gauge shifts preserve the operator's net loss after resolution,
which depends only on price ratios by \eqref{eq:loss}.

To better understand the mechanism, we execute a pair trade step by step in markets with $M=3$, illustrated in Appendix~\ref{app:example}.

\section{Results}
\label{sec:results}

Our evaluation pursues two main goals. First, we show in simulations on synthetic low-order trade flow that LOPMM's loss grows linearly with $M$; we check robustness to the flow parameters through one-at-a-time sweeps,
and we verify that the trader's portfolio and odds are not hurt under
the new design. Second, we replay historical Kalshi order flow through LOPMM and
show that the results survive under realistic trade sequences.

\subsection{Loss scaling with M}
\label{subsec:sim}

\paragraph{Setup.} We generate an opening belief $P^0$ and a latent truth process $P^{(t)}$ in
canonical coordinates (refer to Appendix~\ref{app:theta} for definition on log-linear parameters). At each timestamp $t$, the latent process receives an update in a
$\theta$-coordinate. A trader's information is modeled by selecting the leg
set $S$ whose natural parameters have just moved, and marginalizing onto $S$ to
determine the resulting trade on $S$. Then, an operator takes the full-trade and 
executes it on the market following the mechanics stated in \S~\ref{sec:trade-mechanics}.

This is to make sure the latent distribution has high order changes. In the case of
lower independence, LOPMM achieves an even lower loss, matching the lower bound of 
cost of running $M$ base markets for fully independent base markets.

\begin{figure}[t]
  \centering
  \begin{subfigure}[t]{0.48\linewidth}
    \centering
    \includegraphics[width=\linewidth]{figures/paper_net_loss_pyramid_dec1.25_M20.pdf}
    \caption{Market maker net loss versus number of base events $M$.}
    \label{fig:sim-scaling}
  \end{subfigure}
  \hfill
  \begin{subfigure}[t]{0.48\linewidth}
    \centering
    \includegraphics[width=\linewidth]{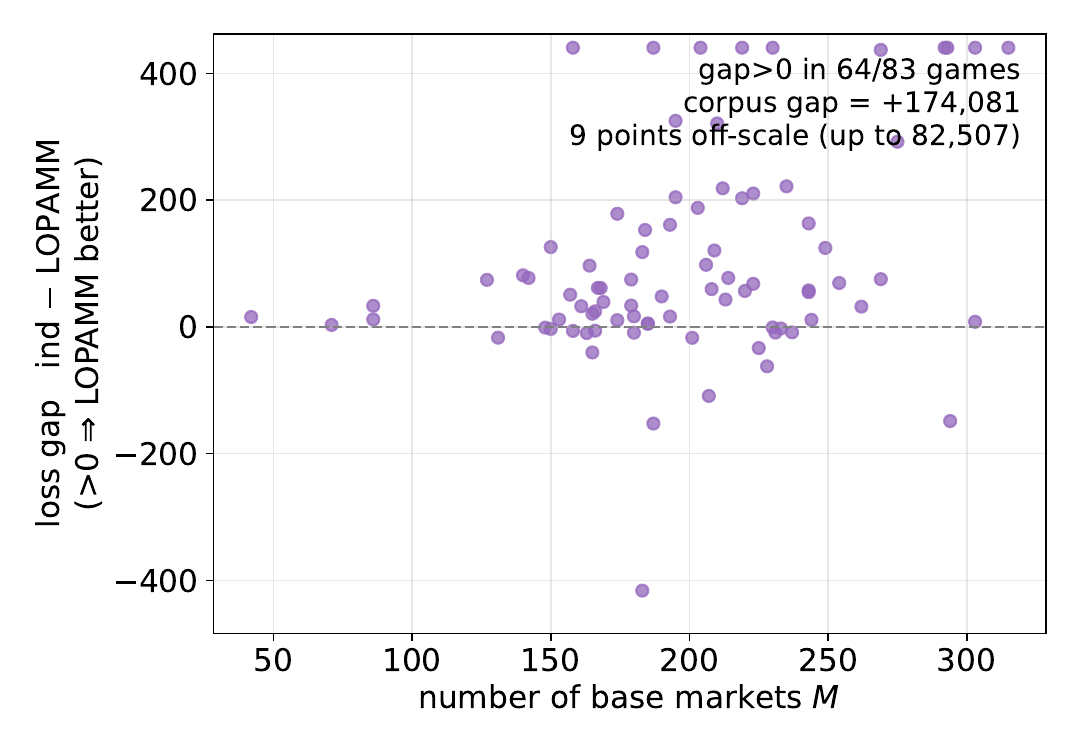}
    \caption{Per-game loss gap against base-market count $M$; positive means LOPMM
    loses less.}
    \label{fig:replay-gap}
  \end{subfigure}
\end{figure}

\paragraph{Market maker loss.} Figure~\ref{fig:sim-scaling} reports operator loss against $M$. 
We compare LOPMM against two references. First one, what we call the independent baseline, 
which runs one LMSR per leg-set, but there is no hierarchy in the parametrization; every LMSR is stand-alone (\emph{ind}). 
Second one, a base-market-only execution that accepts only base market 
requests and discards the parlay trades (\emph{base}). LOPMM's mean loss grows as $O(M)$
across the sweep, tracking the base-market cost, while the independent design 
separates from both within a handful of legs. 

The low-order flow is governed by the two parameters of
Assumption~\ref{ass:low-order}, set to $k=6$, $\rho=0.8$, with liquidity $b=10$, and
$B_\theta=5$ that is the bound on the magnitude of the information received
in higher orders. For each $M$ we run $40$ seeds and report
the mean net loss. The whiskers span the full min--max range over seeds, confirming a
worst-case quadratic bound empirically. The mean itself is markedly flatter, showing a linear growth.

\begin{figure}[t]
  \centering
  \includegraphics[width=\linewidth]{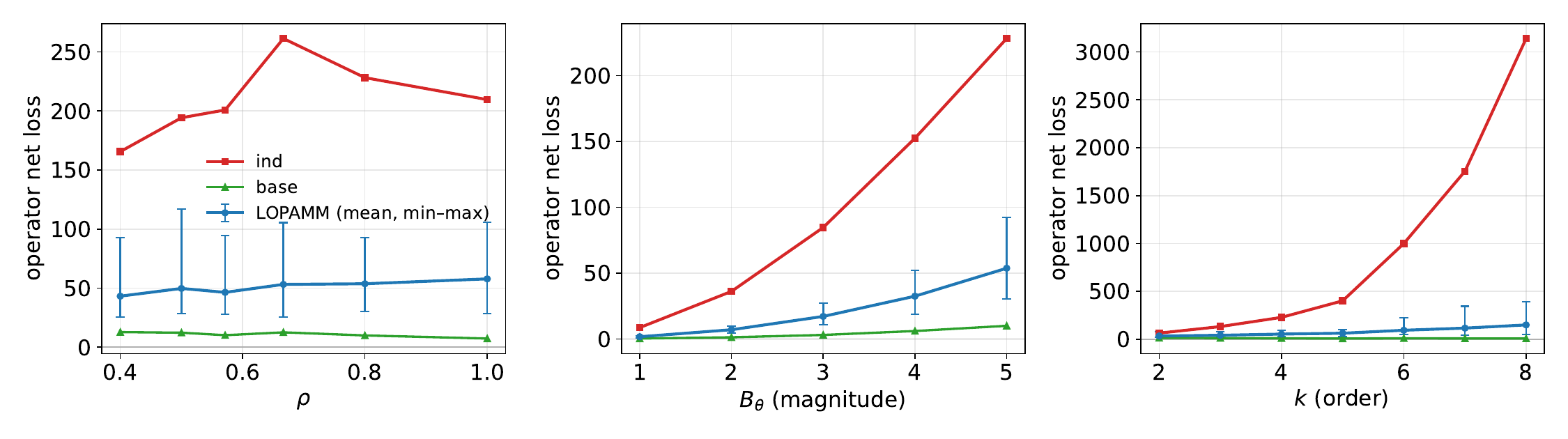}
  \caption{Robustness of loss scaling to the three flow parameters of
  Assumption~\ref{ass:low-order}, sweeping one at a time about the operating
  point ($k=6$, $\rho=0.8$, $B_\theta=5$) at fixed $M=10$ and $b=10$.}
  \label{fig:robustness}
\end{figure}

\paragraph{Robustness.} Figure~\ref{fig:robustness} sweeps each of the three low-order flow parameters
about the operating point, holding the rest fixed, repeating each point for $40$ seeds. Across all three, LOPMM tracks
the base-market floor closer while the independent design
sits far above.
The magnitude parameter $B_\theta$ drives loss upward for every
design, as it adds displacement the operator must finance. It has a
milder effect on LOPMM compared to the independent design. The
decay $\rho$ moves LOPMM only weakly, confirming that the mechanism does not rely on a
sharp decay to stay cheap. The order cap $k$ is the sharpest separator. The
baseline explodes as $k$ grows.

\begin{table}[t]
\centering
\caption{Market maker incremental parlay loss by order,
for ind and LOPMM on replayed Kalshi flow.}
\label{tab:replay-order}
\small
\setlength{\tabcolsep}{4pt}
\begin{tabular}{l rrrrrrr r}
\toprule
$|S|$ & 2 & 3 & 4 & 5 & 6 & 7 & 8 & all \\
\midrule
ind loss   & $16{,}670.9$ & $33{,}324.9$ & $51{,}664.4$ & $43{,}646.1$ & $21{,}870.1$ & $6{,}130.5$ & $743.4$ & $174{,}050$ \\
LOPMM loss & $-24.9$ & $-27.2$ & $2.9$ & $18.6$ & $1.5$ & $-2.0$ & $0.1$ & $-31$ \\
\bottomrule
\end{tabular}
\end{table}

\paragraph{Traders' expected gain}
Next, we confirm that traders received fair odds under LOPMM against the
independent-LMSR pricing. We hold the environment fixed and compare, trade for
trade, the deal a trader realizes under both designs. 
We instantiate a low-order trade flow using same parameters as loss experiment and set
$M=15$ base events and $k=6$ averaged over $100$ latent-belief seeds. Both
designs see the identical flow. For each
trader under each design we record the premium paid and the contracts
acquired on every touched book, and score them against the terminal joint to obtain an expected payout $\bar v$. The trader's effective price is $p_{\mathrm{eff}}=\bar v/\pi$.
Our paired unit is the per-trader ratio
$q=p_{\mathrm{eff}}^{\mathrm{LOPMM}}/p_{\mathrm{eff}}^{\mathrm{ind}}$, stratified by
order $|S|$; $q\approx 1$ means the two designs price the same trade equivalently.

Table~\ref{tab:trader-price} in Appendix~\ref{app:trader-price} reports $q$ by order. LOPMM's saving is a redistribution of the odds.

\subsection{Loss on real order flow}
\label{subsec:kalshi}
\paragraph{Setup.} 
Our final experiment replays actual NBA same-game parlay order flow from Kalshi.
We use $83$ games carrying a meaningful number of traded parlays, spanning
$M\in[42,315]$ base markets per game and parlay orders up to $8$ legs. For each
game we take the time-ordered tape of every parlay trade together with the base
trades on those legs that appear in some parlay, and we settle at the
realized game outcomes. A leg resolves per the official result and a
parlay's conjunction resolves to $1$ iff every leg matches its stated side.

The recorded Kalshi trade flow gives executed \emph{prices} for each trade. A base trade at price $p$ targets the leg quote at distribution
$[1-p,\,p]$. A parlay trade on leg-set $S$ at price $p$ defines a single target
distribution $\tau_S$ over $S$'s $2^{|S|}$ outcomes. The reference is the product
of the current base-leg quotes, 
with the traded conjunction pinned to $p$ and the remaining outcomes rescaled
proportionally.
This one $\tau_S$
is projected onto every sub-market $S'\subseteq S$ and fed to \emph{both} designs,
so they face identical inputs and differ only in how a quote is represented. Each
trade is then executed by the same bottom-up, buys-only sweep used in the
synthetic experiments, at a common liquidity $b=10$. We report the
\emph{incremental parlay loss}: total loss minus the loss of the base-markets-only
book, i.e.\ the cost of offering the parlay product at all. Parlays with more than
$8$ legs are dropped for computation efficiency, which removes $185$ of
$29{,}046$ parlay trades ($0.6\%$).

\paragraph{Results.}
Table~\ref{tab:replay-order} resolves the incremental loss by parlay order. Across
the corpus the independent books lose $\$174{,}050$, whereas LOPMM's loss is within $\$28$ of zero at every
order and sums to $-\$31$. Offering the full parlay book therefore costs LOPMM
essentially nothing relative to the base-markets-only floor, while the same product costs the
independent design an order of magnitude more than the floor itself.

Figure~\ref{fig:replay-gap} plots the per-game loss gap against the number of base
markets $M$, where a positive gap means LOPMM loses less on that game. The gap is
positive in $64$ of $83$ games and does not shrink as $M$ grows
(Pearson $r=-0.05$, $p=0.64$). Moreover, the comparison is strongly asymmetric. The $64$
games favouring LOPMM contribute $+\$175{,}134$, while the $19$ games favouring
the independent books contribute only $-\$1{,}053$ in total.

\section{Analysis: Computation and Loss}
\label{sec:analysis}

In this section, we compare three designs on the same $M$ base events. The single
\emph{joint} LMSR over the $2^M$ atoms
$z\in\{0,1\}^M$~\cite{hanson2003combinatorial}, the \emph{independent} LMSR
books, and
LOPMM. The two quantities we care about are the operator's loss at resolution 
and the arithmetic needed to quote or execute an
order-$m$ parlay.


\subsection{Accounting}
\label{sec:acc}

\paragraph{Opening state.}
Without losing generality, we assume LOPMM opens every block at zero, $r^{(S)}\equiv0$ for all $S$, so every
book opens at the uniform quote $\pi^0_S=u_S:=2^{-|S|}$ (for an informative
opening, one could replace $\log2^{|S|}$ by $\log(1/\pi^0_S(\boldsymbol\omega^\star_S))$
throughout.)

\paragraph{Paid and unpaid moves.}
A sub-trade on book $S'$ changes
only the block $r^{(S')}$, by $\Delta r^{(S')}\in\mathbb R^{\Omega_{S'}}$,
and is priced by book $S'$'s cost \eqref{eq:cost-design} at the state of
that book when the sweep reaches it, i.e.\ with the lower-order sub-trades
of the same full-trade already propagated into its aggregates
\eqref{eq:Q-agg}. The trader pays
$C^{(S')}(\mathbf Q^{(S')}+\Delta r^{(S')})-C^{(S')}(\mathbf Q^{(S')})$ and
receives $\Delta r^{(S')}_{\boldsymbol\omega}$ contracts on each pattern
$\boldsymbol\omega\in\Omega_{S'}$. We call this a \emph{paid move} on $S'$.
A change of $\mathbf Q^{(S)}$ caused by a sub-trade on some $T\subsetneq S$
is an \emph{unpaid move} on $S$; no contracts are issued and nothing is
paid. At resolution the base events land on an atom
$z^\star\in\{0,1\}^M$; book $S$ pays $\$1$ per contract on the pattern
$\boldsymbol\omega^\star_S:=z^\star|_S$, and the operator's loss is payout
minus revenue.

Write $\mathcal T_S$ for the time-ordered set of paid moves on book $S$,
$n_S$ for the number of full-trades whose leg-set is exactly $S$, and, as in
Assumption~\ref{ass:low-order}, $N_m=\sum_{|S|=m}n_S$. For
$t\in\mathcal T_S$ write $\pi^{-,t}_S,\pi^{+,t}_S$ for the quote of book $S$
immediately before and after it.

For higher-order books we need the size of a single paid move.

\begin{assumption}[Bounded displacement]
\label{ass:disp}
There is a constant $B$ such that for every paid move $t$ on every book $S$
and every $\boldsymbol\omega\in\Omega_S$,
$\log\pi^{+,t}_S(\boldsymbol\omega)-\log\pi^{-,t}_S(\boldsymbol\omega)\le B$.
\end{assumption}

If every quote the operator posts, before and after each sub-trade, lies
in $[\epsilon,1-\epsilon]$ (tick $\epsilon$; one cent on Kalshi), then
Assumption~\ref{ass:disp} holds with $B=\log\frac{1-\epsilon}{\epsilon}$. 
$B$ depends neither on $M$ nor on the order of the book.

\subsection{Worst-case loss}
\label{sec:worst}

\begin{restatable}[Worst-case loss]{theorem}{thmworst}
\label{thm:worst}
Under the low-order assumption~\ref{ass:low-order} and bounded displacement~\ref{ass:disp}, a flow from partially informed and honest traders who report $\tau^\star_S$ on $S$, the worst case loss for the operator is bounded by $O(M^2)$.
\end{restatable}

The proof is in Appendix~\ref{proof:thmworst}


\subsection{Honest, informed flow on a fixed truth}
\label{sec:informed}

Throughout this subsection the latent truth $P^\star$ is fixed and every
arriving trader is partially informed and honest, so, on $S$ reports $\tau^\star_S$. Let
$\mathcal D$ be the set of traded leg-sets and
$\mathrm{cov}:=\{T\neq\emptyset: T\subseteq S\text{ for some }S\in\mathcal D\}$
the leg-sets \emph{covered} by the flow; $\mathrm{cov}_m$ and
$\mathrm{cov}_{\ge2}$ denote the covered sets of order $m$ and of order at
least two. Define the \emph{truth blocks} recursively over the subset
lattice,
\begin{equation}
  \label{eq:truth-blocks}
  r^{\downarrow(T)}_{\boldsymbol\omega}\;:=\;b\log\tau^\star_T(\boldsymbol\omega)
  -\!\!\sum_{\emptyset\neq V\subsetneq T}\! r^{\downarrow(V)}_{\boldsymbol\omega|_V},
  \qquad \boldsymbol\omega\in\Omega_T,
\end{equation}
and the \emph{extrapolated quote} of book $T$ as the LMSR quote formed from
the truth blocks of its proper subsets alone,
\begin{equation}
  \label{eq:extrap}
  \pi^{\downarrow}_T(\boldsymbol\omega)\;\propto\;
  \exp\Big(\tfrac1b\!\sum_{\emptyset\neq V\subsetneq T}\! r^{\downarrow(V)}_{\boldsymbol\omega|_V}\Big),
  \qquad
  \pi^{\downarrow}_{\{i\}}=u_1,\qquad
  \pi^{\downarrow}_{\{i,j\}}=\tau^\star_i\otimes\tau^\star_j .
\end{equation}
$\pi^\downarrow_T$ is what book $T$ quotes when every lower-order book is at
its truth and the order-$|T|$ block is still zero; it depends on $P^\star$
only. Both objects are defined up to a uniform shift of each block, which
by Proposition~\ref{prop:identity} affects nothing.

\begin{restatable}[Exact loss on honest, informed flow]{theorem}{thminformed}
\label{thm:informed}
Under the standing assumptions of this subsection and any arrival order:
\begin{enumerate}[label=(\alph*),leftmargin=*]
  \item (Path independence.) For every resolution $z^\star$,
  \begin{equation}
    \label{eq:informed-loss}
    L(z^\star)\;=\;\sum_{i\in\mathrm{cov}_1} b\log\big(2\,\tau^\star_i(z^\star_i)\big)
    \;+\;\sum_{T\in\mathrm{cov}_{\ge2}} b\log
    \frac{\tau^\star_T(z^\star|_T)}{\pi^{\downarrow}_T(z^\star|_T)} ,
  \end{equation}
  which depends on the flow only through $\mathrm{cov}$.
  \item (Worst case.) $L(z^\star)\le bM\log2+b\sum_{T\in\mathrm{cov}_{\ge2}}\beta_T$,
  where
  $\beta_T:=\max_{\boldsymbol\omega}\log\big(\tau^\star_T(\boldsymbol\omega)/\pi^\downarrow_T(\boldsymbol\omega)\big)\le B$
  under Assumption~\ref{ass:disp}.
  \item (Expected loss.) With $z^\star\sim P^\star$,
  \begin{equation}
    \label{eq:informed-expected}
    \mathbb E\,L\;=\;b\sum_{i\in\mathrm{cov}_1}\mathrm{KL}(\tau^\star_i\,\|\,u_1)
    \;+\;b\sum_{T\in\mathrm{cov}_{\ge2}}\mathrm{KL}\big(\tau^\star_T\,\|\,\pi^{\downarrow}_T\big).
  \end{equation}
\end{enumerate}
\end{restatable}

The proof is in Appendix~\ref{proof:thminformed}.

Theorem~\ref{thm:informed} is the precise sense in which LOPMM
``shares information''. A book of order $m$ is charged only for the part
of $\tau^\star_T$ that its lower-order books do not already imply.

\begin{restatable}[Worst case on informed and honest reporting flow]{corollary}{corinformedworst}
\label{cor:informed-worst}
Under the standing assumptions of this subsection, Assumption~\ref{ass:disp},
and under low-order flow, the worst case loss for informed and honest reporting 
flow is bounded by $O(M^2)$.
\end{restatable}
The proof is in Appendix~\ref{proof:corinformedworst}.

In expectation the charge per covered set is an information quantity.
A hypothesis on the truth is that the
dependence among base events is local, so that no leg carries more than a
bounded amount of incremental interaction information. Same-game parlays
are the motivating case, since legs of different games are independent and
a game has boundedly many legs.

\begin{assumption}[Local dependence]
\label{ass:local}
There is a constant $\kappa$ such that for every leg $i\in[M]$,
$\displaystyle\sum_{T\in\mathrm{cov}_{\ge2}:\,T\ni i}
  \mathrm{KL}\big(\tau^\star_T\,\|\,\pi^\downarrow_T\big)\;\le\;\kappa .$
\end{assumption}

\begin{theorem}[Expected loss]
\label{thm:expected}
Under the assumptions of Theorem~\ref{thm:informed} and
Assumption~\ref{ass:local}, for every informed and honest reporting flow,
\begin{equation}
  \label{eq:expected-linear}
  \mathbb E\,L\;\le\;bM\log2+\tfrac12\,b\,\kappa\,M\;=\;O(M).
\end{equation}
\end{theorem}

\begin{proof}
Start from \eqref{eq:informed-expected}. Each base term is
$\mathrm{KL}(\tau^\star_i\|u_1)\le\log2$ and there are at most $M$ of them.
Every $T\in\mathrm{cov}_{\ge2}$ contains at least two legs, so
\[
  \sum_{T\in\mathrm{cov}_{\ge2}}\mathrm{KL}(\tau^\star_T\|\pi^\downarrow_T)
  \;\le\;\frac12\sum_{i\in[M]}\;\sum_{T\in\mathrm{cov}_{\ge2}:\,T\ni i}
  \mathrm{KL}(\tau^\star_T\|\pi^\downarrow_T)\;\le\;\frac12\,\kappa M
\]
by Assumption~\ref{ass:local}. Multiply by $b$ and add.
\end{proof}

Theorem~\ref{thm:expected} does not need Assumption~\ref{ass:low-order}.
In expectation, a traded pair of independent legs costs
nothing, and traded pairs of dependent legs are limited by the truth, not
by the flow. 

\begin{proposition}[Comparison with the joint and independent designs]
\label{prop:compare}
Let $P^\star$ be fixed and the flow informed and honest reporting as above.
\begin{enumerate}[label=(\roman*),leftmargin=*]
  \item \emph{Joint LMSR.} For any flow whose terminal joint quote is
  $P^\star$, $L^{\mathrm{joint}}=b\log(2^MP^\star(z^\star))\le bM\log2$ and
  $\mathbb EL^{\mathrm{joint}}=b\,\mathrm{KL}(P^\star\|u_M)=b\big(M\log2-H(P^\star)\big)$.
  \item \emph{Independent books.}
  $L^{\mathrm{ind}}=\sum_{S\in\mathcal D}b\log\big(2^{|S|}\tau^\star_S(z^\star|_S)\big)$
  and $\mathbb EL^{\mathrm{ind}}=b\sum_{S\in\mathcal D}\mathrm{KL}(\tau^\star_S\|u_S)$.
\end{enumerate}
\end{proposition}

\begin{proof}
(i) One LMSR with uniform opening; \eqref{eq:loss} and \eqref{eq:loss-kl}.
(ii) Each independent book is one LMSR whose first trade moves it
from $u_S$ to $\tau^\star_S$ and whose later trades are no-ops.
\end{proof}


\subsection{Computation}
\label{sec:compute}

\begin{proposition}[Arithmetic per quote and per trade]
\label{prop:compute}
Let $S$ have order $m$.
\begin{enumerate}[label=(\roman*),leftmargin=*]
  \item Quoting book $S$ requires the $2^m$ aggregates \eqref{eq:Q-agg},
  each a sum of at most $2^m-1$ stored entries, followed by one softmax:
  $\Theta(4^m)$ operations, or $\Theta(2^m\nu_S)$ with $\nu_S$ the number of
  non-zero blocks $r^{(T)}$, $T\subseteq S$.
  \item A full-trade on $S$ marginalises the target onto every
  $S'\subseteq S$ ($O(2^m)$ each, $O(4^m)$ in total) and performs one
  quote-and-update on each $S'$: $\sum_{j=1}^m\binom mj\Theta(4^j)=\Theta(5^m)$
  operations, touching $\sum_{j=1}^m\binom mj2^j=3^m-1$ stored entries.
  \item Neither cost depends on $M$. With a sparse store the state size is
  $\sum_{T\in\mathrm{cov}}2^{|T|}\le\sum_{S\in\mathcal D}3^{|S|}$ rather than
  the dense $3^M-1$ of \eqref{eq:param-count}.
\end{enumerate}
By contrast, the joint LMSR needs $\Theta(2^M)$ operations per quote or
trade when the $2^M$ atoms are stored explicitly, and its prices are
\#P-hard in succinct form even for two-event
conjunctions~\cite{chen2008complexity}; an independent book of order $m$
costs $\Theta(2^m)$ per quote or trade.
\end{proposition}

\begin{proof}
(i) and (ii) count the terms of \eqref{eq:Q-agg}, \eqref{eq:cost-design} and
\eqref{eq:route} over the subsets processed by the sweep;
$\sum_j\binom mj4^j=5^m-1$ and $\sum_j\binom mj2^j=3^m-1$ by the binomial
theorem. (iii) follows from Theorem~\ref{thm:informed}(a) and the routing
rule, which touch only covered blocks.
\end{proof}


\begin{table}[t]
\centering
\caption{The three designs comparison table. 
}
\label{tab:designs}
\small
\begin{tabular}{lccc}
\toprule
 & joint LMSR & independent books & LOPMM \\
\midrule
state size & $2^M$ & $\sum_{S\in\mathcal D}2^{|S|}$ & $\sum_{T\in\mathrm{cov}}2^{|T|}$ \\
quote / trade, order $m$ & $\Theta(2^M)$ & $\Theta(2^m)$ & $\Theta(4^m)$ / $\Theta(5^m)$ \\
loss, low-order flow for LOPMM & $O(M)$ & $O(M^2)$ & $O(M^2)$ \\
expected loss, informed, honest flow & $O(M)$ & $O(M^2)$ & $O(M)$ \\
\bottomrule
\end{tabular}
\end{table}

\section{Conclusion}
\label{sec:conclusion}

The parlay family is exponentially large, but the observed flow through it is concentrated at low orders. LOPMM exploits this structure by coupling low-order books through a shared hierarchical parameterization, concentrating computation and liquidity where trading actually occurs. Our experiments suggest that this substantially reduces the additional cost of supporting parlay markets relative to independent books, while retaining computation that depends on parlay order rather than the size of the global market. The analysis complements these results by identifying structured-flow settings in which the market maker’s loss remains controlled: under informed-truthful low-order flow, worst case loss is $O(M^2)$, and under local dependence the expected loss is $O(M)$.


\paragraph{Future work.}

The current guarantees rely on assumptions about the information and behavior represented in the trading flow. Extending the mechanism and its analysis to strategic traders, and identifying the additional structure needed for incentive compatibility, is therefore an important direction. In particular, LOPMM is stateful, so without long-horizon incentive compatibility, individually valid trades can interact through the evolving state to create cyclic strategies that repeatedly extract value from the operator. Characterizing and preventing such loop attacks is therefore important for deployment beyond the informed-truthful setting. Assumption~\ref{ass:low-order} is grounded empirically. We measure the low-order structure directly rather than positing it. Deriving this structure endogenously would make the guarantees more self-contained. Another natural extension is a fee mechanism layered on the hierarchical parameterization, designed to preserve the desirable loss properties while allowing the operator to extract a margin.


\bibliographystyle{splncs04}
\bibliography{references}

\appendix

\section{Proofs}
\label{app:proofs}
\subsection{Proposition and lemma used in other proofs}
\label{proof:others}

\begin{restatable}[Loss identity; gauge neutrality]{proposition}{propidentity}
\label{prop:identity}
Fix $z^\star$. A paid move $t$ on book $S$ costs the operator
\begin{equation}
  \label{eq:paid-loss}
  \ell_t
  \;=\;\Delta r^{(S),t}_{\boldsymbol\omega^\star_S}
      -\Big(C^{(S)}(\mathbf Q^{(S)}+\Delta r^{(S),t})-C^{(S)}(\mathbf Q^{(S)})\Big)
  \;=\;b\log\frac{\pi^{+,t}_S(\boldsymbol\omega^\star_S)}
                 {\pi^{-,t}_S(\boldsymbol\omega^\star_S)} ,
\end{equation}
and the operator's total loss is
\begin{equation}
  \label{eq:total-loss}
  L(z^\star)\;=\;\sum_{\emptyset\neq S\subseteq[M]}\;\sum_{t\in\mathcal T_S}\ell_t .
\end{equation}
Adding a constant $c$ to $\Delta r^{(S),t}$ (the buy-only gauge of
\S\ref{sec:trade-mechanics}) changes neither $\ell_t$ nor any quote.
\end{restatable}

\begin{proof}
Only paid moves issue contracts or collect payment, so the payout is
$\sum_S\sum_{t\in\mathcal T_S}\Delta r^{(S),t}_{\boldsymbol\omega^\star_S}$
and the revenue is the sum of the payments; subtracting termwise gives the
first equality in \eqref{eq:paid-loss} and \eqref{eq:total-loss}. For the
second, with $\mathbf Q=\mathbf Q^{(S)}$ and $\Delta=\Delta r^{(S),t}$,
\eqref{eq:cost-design} gives
\[
  C^{(S)}(\mathbf Q+\Delta)-C^{(S)}(\mathbf Q)
  = b\log\frac{\sum_{\boldsymbol\omega'}e^{(Q_{\boldsymbol\omega'}+\Delta_{\boldsymbol\omega'})/b}}
              {\sum_{\boldsymbol\omega'}e^{Q_{\boldsymbol\omega'}/b}},
  \qquad
  b\log\frac{\pi^{+}(\boldsymbol\omega^\star)}{\pi^{-}(\boldsymbol\omega^\star)}
  = \Delta_{\boldsymbol\omega^\star}
   - b\log\frac{\sum_{\boldsymbol\omega'}e^{(Q_{\boldsymbol\omega'}+\Delta_{\boldsymbol\omega'})/b}}
               {\sum_{\boldsymbol\omega'}e^{Q_{\boldsymbol\omega'}/b}} .
\]
Replacing $\Delta$ by $\Delta+c$ adds $c$ to $\Delta_{\boldsymbol\omega^\star}$
and $c$ to the cost difference, so $\ell_t$ is unchanged; a uniform shift
of $r^{(S)}$ shifts $\mathbf Q^{(S'')}$ uniformly for every $S''\supseteq S$,
which the softmax in \eqref{eq:cost-design} ignores.
\end{proof}

\begin{lemma}[Base books]
\label{lem:base}
For $|S|=1$ and any flow,
$\sum_{t\in\mathcal T_S}\ell_t
 = b\log\big(2\,\pi^{\mathrm{fin}}_S(\boldsymbol\omega^\star_S)\big)\le b\log2$,
independently of $N_1$.
\end{lemma}

\begin{proof}
$\mathbf Q^{(S)}=r^{(S)}$ changes only in paid moves, so consecutive
$\pi^{+,t}_S$ and $\pi^{-,t'}_S$ coincide and \eqref{eq:paid-loss}
telescopes from $\pi^0_S=\tfrac12$ to $\pi^{\mathrm{fin}}_S$.
\end{proof}

\begin{lemma}[One paid move]
\label{lem:onemove}
Under Assumption~\ref{ass:disp}, $\ell_t\le bB$ for every paid move $t$.
\end{lemma}

\begin{proof}
Apply Assumption~\ref{ass:disp} at
$\boldsymbol\omega=\boldsymbol\omega^\star_S$ in \eqref{eq:paid-loss}.
\end{proof}

\begin{lemma}[Counting paid moves]
\label{lem:count}
A full-trade on a leg-set of order $m$ issues one sub-trade on each of its
$2^m-1$ non-empty subsets, $m$ of which are base books. Hence
$|\mathcal T_S|=\sum_{S''\supseteq S}n_{S''}$ and
\begin{equation}
  \label{eq:count}
  \sum_{|S|\ge2}|\mathcal T_S|\;=\;\sum_{m\ge2}N_m\,(2^m-m-1).
\end{equation}

Under Assumption~\ref{ass:low-order}, $N_m\le\rho^{m-2}N_2$ for
$2\le m\le k$ and $N_m=0$ for $m>k$, so
\begin{equation}
  \label{eq:crho}
  \sum_{|S|\ge2}|\mathcal T_S|\;\le\;c_{\rho,k}\,N_2,
  \qquad
  c_{\rho,k}:=\sum_{m=2}^{k}\rho^{m-2}\,(2^m-m-1)\;\le\;4\sum_{j=0}^{k-2}(2\rho)^j .
\end{equation}
\end{lemma}

\begin{proof}
The first two claims restate the routing rule of
\S\ref{sec:trade-mechanics}; \eqref{eq:count} sums them over $S''$. Then,
substitute $N_m\le\rho^{m-2}N_2$ into \eqref{eq:count}.
\end{proof}

\subsection{Worst case loss}
\label{proof:thmworst}

\thmworst*

\begin{proof}
For every resolution $z^\star$, split \eqref{eq:total-loss} into base books and the rest. At most $M$ base
books are ever touched, and each contributes at most $b\log2$ by
Lemma~\ref{lem:base}, whatever $N_1$. Every paid move on a book of order at
least two contributes at most $bB$ by Lemma~\ref{lem:onemove}, and there
are $\sum_{m\ge2}N_m(2^m-m-1)$ of them by \eqref{eq:count}, so we have
\begin{equation}
  \label{eq:worst-general}
  L(z^\star)\;\le\;bM\log 2\;+\;bB\sum_{m\ge2}N_m\,(2^m-m-1).
\end{equation}
Under the low-order assumption~\ref{ass:low-order} and bounded displacement~\ref{ass:disp},
\eqref{eq:crho} bounds the count by $c_{\rho,k}N_2$, and $N_2\le\nu\binom M2\le\tfrac12\nu M^2$.
Therefore, we get
\begin{equation}
  \label{eq:worst-low}
  L(z^\star)\;\le\;bM\log 2\;+\;bB\,c_{\rho,k}\,N_2
  \;\le\;bM\log2+\tfrac12\,bB\,c_{\rho,k}\,\nu\,M^2\;=\;O(M^2).
\end{equation}

\end{proof}


\subsection{Exact Loss on honest, informed flow}
\label{proof:thminformed}
\begin{lemma}[One-shot block updates]
\label{lem:oneshot}
Under the standing assumptions of \S\ref{sec:informed} and any arrival order,
each block $r^{(T)}$ with $T\in\mathrm{cov}$ moves, up to uniform shifts,
exactly once, at the first full-trade whose leg-set contains $T$. At that
moment every $r^{(V)}$ with $V\subsetneq T$ equals $r^{\downarrow(V)}$ up to
a uniform shift, and afterwards $r^{(T)}\equiv r^{\downarrow(T)}$. Blocks
with $T\notin\mathrm{cov}$ never move, and every other paid move is a
uniform shift.
\end{lemma}

\begin{proof}
We show, by induction on $|T|$ and within an order by time, that
whenever a sub-trade on $T$ is executed the new block satisfies
$r^{(T)}\equiv r^{\downarrow(T)}$ (equality up to a uniform shift). A
sub-trade on $T$ occurs inside a full-trade on some $S\supseteq T$ whose
target is $\tau^\star_S$; since all traders report marginals of the same
$P^\star$, the order-$|T|$ target in \eqref{eq:target} is
$g_T=b\log\tau^\star_T$. The sweep is bottom-up, so every $V\subsetneq T$
has been processed earlier in the same full-trade, and by the inductive
hypothesis $r^{(V)}\equiv r^{\downarrow(V)}$ at that moment. The routing
rule \eqref{eq:route} then gives
$r^{(T)}\leftarrow b\log\tau^\star_T-\sum_{V\subsetneq T}r^{(V)}
\equiv r^{\downarrow(T)}$ by \eqref{eq:truth-blocks}; for $|T|=1$ the sum
is empty. If $T$ had been set before, its old value was
$\equiv r^{\downarrow(T)}$ as well, so the move is a uniform shift, which by
Proposition~\ref{prop:identity} changes no quote and costs nothing. If $T$
is set for the first time, $r^{(T),\mathrm{old}}=0$ and by \eqref{eq:Q-agg}
the pre-move quote of book $T$ is exactly $\pi^\downarrow_T$. A block
$r^{(T)}$ is touched only by sub-trades on $T$, which occur only inside
full-trades on supersets of $T$, so blocks outside $\mathrm{cov}$ never
move.
\end{proof}

\thminformed*

\begin{proof}
(a) By Lemma~\ref{lem:oneshot}, the only paid moves with nonzero cost are the
first move of each block $r^{(T)}$, $T\in\mathrm{cov}$; uniform shifts cost
nothing by Proposition~\ref{prop:identity}.
For $|T|\ge2$ the move takes book $T$ from
$\pi^\downarrow_T$ to $\tau^\star_T$, so by \eqref{eq:paid-loss} it costs
$b\log(\tau^\star_T(z^\star|_T)/\pi^\downarrow_T(z^\star|_T))$. For $|T|=1$
Lemma~\ref{lem:base} applies with $\pi^{\mathrm{fin}}_{\{i\}}=\tau^\star_i$.
Substituting into \eqref{eq:total-loss} gives \eqref{eq:informed-loss}.

(b) Bound each base term by $b\log2$ (at most $M$ of them) and each
higher-order term by $b\beta_T$; $\beta_T\le B$ is
Assumption~\ref{ass:disp} applied to the single non-trivial move on $T$.
If $\tau^\star_T=\pi^\downarrow_T$ the routing yields $\Delta r^{(T)}$
constant, hence $\beta_T=0$.

(c) Under $z^\star\sim P^\star$ the pattern $z^\star|_T$ has law
$\tau^\star_T$ by \eqref{eq:true-marginal}, so the expectation of each term
of \eqref{eq:informed-loss} is $b$ times the KL divergence.
For a pair, $\pi^\downarrow_{\{i,j\}}=\tau^\star_i\otimes\tau^\star_j$ by
\eqref{eq:extrap}, and
$\mathrm{KL}(\tau^\star_{ij}\|\tau^\star_i\otimes\tau^\star_j)=I^\star(i;j)$.
\end{proof}

\subsection{Worst case Loss on honest, informed flow}
\label{proof:corinformedworst}

\corinformedworst*

\begin{proof}
Under the standing assumptions of this subsection, Assumption~\ref{ass:disp},
and any arrival order, combine Theorem~\ref{thm:informed}(b) with the count of subsets of order at
least two of each $S\in\mathcal D$, so we get
\begin{equation}
  \label{eq:cov-count}
  L(z^\star)\;\le\;bM\log2+bB\,|\mathrm{cov}_{\ge2}|,
  \qquad
  |\mathrm{cov}_{\ge2}|\;\le\;\sum_{m\ge2}|\mathcal D_m|\,(2^m-m-1),
\end{equation}
and under low-order flow (Assumption~\ref{ass:low-order}) the right-hand
side is at most $c_{\rho,k}N_2\le\tfrac12c_{\rho,k}\nu M^2$, since
$|\mathcal D_m|\le N_m$; so $L=O(M^2)$.
\end{proof}


\section{Worked Example: a pair trade with $M=3$}
\label{app:example}

Take $M=3$, $b=1$, and a fresh book ($r^{(S)}\equiv 0$ for all $S$), so every
market quotes uniformly. A trader targets the pair $S=\{1,2\}$ with
\[
  \tau^{\mathrm{new}}_{\{1,2\}}(00,01,10,11)=(0.1,\;0.2,\;0.2,\;0.5),
\]
whose leg marginals are $\tau^{\mathrm{new}}_{\{1\}}=\tau^{\mathrm{new}}_{\{2\}}=(0.3,0.7)$.
The full-trade has three sub-trades, executed bottom-up.

\paragraph{Order 1.} By \eqref{eq:route}, $q^{(\{k\})}_\omega=\log\tau^{\mathrm{new}}_{\{k\}}(\omega)+c_k$
for $k=1,2$. The smallest buy-only gauge is $c_k=-\log 0.3$, giving
$q^{(\{k\})}=(0,\log\tfrac73)\approx(0,0.847)$. Each base market now quotes
$(0.3,0.7)$ at cost $\log\tfrac{10}{3}-\log 2=\log\tfrac53\approx0.511$.

\paragraph{Order 2.} With the singletons fixed, \eqref{eq:route} gives
$r^{(\{1,2\})}_{\omega_1\omega_2}=\log\tau^{\mathrm{new}}_{\{1,2\}}(\omega_1\omega_2)
-q^{(\{1\})}_{\omega_1}-q^{(\{2\})}_{\omega_2}+c_{12}$. The buy-only gauge sets
$c_{12}=\log\tfrac{35}{3}$, so
\[
  r^{(\{1,2\})}=\bigl(\log\tfrac76,\;0,\;0,\;\log\tfrac{15}{14}\bigr)\approx(0.154,0,0,0.069).
\]
The aggregate \eqref{eq:Q-agg} is then
$e^{Q^{(\{1,2\})}}=(\tfrac76,\tfrac73,\tfrac73,\tfrac{35}{6})\propto(0.1,0.2,0.2,0.5)$,
so the pair market quotes the target exactly. This sub-trade costs
$\log\tfrac{70}{6}-\log\tfrac{100}{9}\approx0.049$, and the full-trade costs
$\approx 1.07$ in total. 

After these three trades, all other blocks remain zero, yet their markets reprice
because $Q^{(S)}$ aggregates every consistent lower-order block. The state of all 
blocks and their quotes after the trade execution are in Figure~\ref{fig:example-state}.

\begin{figure*}[h]
\centering
\small
\renewcommand{\arraystretch}{1.25}
\begin{tabular}{@{}l l l l@{}}
\toprule
Market $S$ & Stored block $r^{(S)}$ & Aggregate $e^{Q^{(S)}}$ & Quote $\pi^{(S)}$ \\
\midrule
\rowcolor{gray!15}
$\{1\}$   & $(0,\ \log\tfrac73)$ & $(1,\ \tfrac73)$ & $(0.3,\ 0.7)$ \\
\rowcolor{gray!15}
$\{2\}$   & $(0,\ \log\tfrac73)$ & $(1,\ \tfrac73)$ & $(0.3,\ 0.7)$ \\
$\{3\}$   & $(0,\ 0)$            & $(1,\ 1)$        & $(0.5,\ 0.5)$ \\
\midrule
\rowcolor{gray!15}
$\{1,2\}$ & $(\log\tfrac76,\ 0,\ 0,\ \log\tfrac{15}{14})$
          & $(\tfrac76,\ \tfrac73,\ \tfrac73,\ \tfrac{35}{6})$
          & $(0.1,\ 0.2,\ 0.2,\ 0.5)$ \\
$\{1,3\}$ & $(0,\ 0,\ 0,\ 0)$ & $(1,\ 1,\ \tfrac73,\ \tfrac73)$
          & $(0.15,\ 0.15,\ 0.35,\ 0.35)$ \\
$\{2,3\}$ & $(0,\ 0,\ 0,\ 0)$ & $(1,\ 1,\ \tfrac73,\ \tfrac73)$
          & $(0.15,\ 0.15,\ 0.35,\ 0.35)$ \\
\midrule
$\{1,2,3\}$ & $\mathbf{0}$
          & $(\tfrac76,\tfrac76,\ \tfrac73,\tfrac73,\ \tfrac73,\tfrac73,\ \tfrac{35}{6},\tfrac{35}{6})$
          & $(0.05,0.05,\ 0.1,0.1,\ 0.1,0.1,\ 0.25,0.25)$ \\
\bottomrule
\end{tabular}
\caption{State of all seven markets after the pair trade. Outcomes are listed in binary order
($0,1$; $00,01,10,11$; $000,\dots,111$). Shaded rows are the blocks moved by
the three sub-trades.}
\label{fig:example-state}
\end{figure*}

\section{Experiments and Empirical Results}
\subsection{Log-linear parameters}
\label{app:theta}

For $U\subseteq S$, let $\mathbf 1_U\in\Omega_S$ be the pattern in which
exactly the legs in $U$ land. Any strictly positive $\pi_S\in\Delta(\Omega_S)$
has a unique log-linear form with one coordinate per non-empty $T\subseteq S$,
\begin{equation}
  \label{eq:loglinear}
  \log\pi_S(\mathbf 1_U) \;=\; \log\pi_S(\mathbf 1_\emptyset)
  + \sum_{\emptyset\neq T\subseteq U}\theta_T[\pi_S],
  \qquad U\subseteq S,
\end{equation}
inverted by M\"obius inversion on the subset lattice,
\begin{equation}
  \label{eq:canonical}
  \theta_T[\pi_S] \;=\; \sum_{U\subseteq T}(-1)^{|T|-|U|}\log\pi_S(\mathbf 1_U).
\end{equation}
Here $\theta_{\{i\}}[\pi_S]$ is the log-odds of leg $i$ conditional on the
other legs of $S$ not landing, and $\theta_T[\pi_S]$ for $|T|\ge2$ is the
interaction among the legs of $T$, zero for all $|T|\ge2$ exactly when
$\pi_S$ is a product distribution. 

\subsection{Trader-side Prices}
\label{app:trader-price}
Refer to Table~\ref{tab:trader-price}.
\begin{table}[t]
\centering
\caption{Paired per-trader price ratio
$q=p_{\mathrm{eff}}^{\mathrm{LOPMM}}/p_{\mathrm{eff}}^{\mathrm{ind}}$
by parlay order, 100 seeds.}
\label{tab:trader-price}
\begin{tabular}{rccc}
\toprule
$|S|$ & mean $q$ & $95\%$ CI & in $\pm10\%$ \\
\midrule
1     & 1.000 & $[1.000,1.000]$ & 100\% \\
2     & 0.994 & $[0.990,0.998]$ & 77\% \\
3     & 0.981 & $[0.976,0.986]$ & 62\% \\
4     & 0.971 & $[0.964,0.977]$ & 56\% \\
5     & 0.961 & $[0.954,0.969]$ & 55\% \\
6     & 0.938 & $[0.931,0.946]$ & 52\% \\
\midrule
pool  & 0.978 & $[0.976,0.980]$ & 68\% \\
\bottomrule
\end{tabular}
\end{table}


\section{Supplementary Flow Data}
\label{app:supp-data}
 
Tables~\ref{tab:flow_nba_app}, \ref{tab:flow_agg_app}
and~\ref{tab:flow_7d_app} give the full leg-count breakdowns of the three
Kalshi datasets. ``Cum.'' is the
cumulative share of parlay trades up to and including that row.
 
\begin{table}[h]
\centering
\caption{Traded NBA same-game parlays on Kalshi by leg count (April to June
2026, $M = 2,907$ base markets).}
\label{tab:flow_nba_app}
\begin{tabular}{rrrr}
\toprule
legs & markets & trades & cum.\ (\%) \\
\midrule
2 & 6{,}250 & 28{,}505 & 95.7 \\
3 & 372 & 409 & 97.1 \\
4 & 176 & 210 & 97.8 \\
5 & 140 & 168 & 98.4 \\
6 & 109 & 114 & 98.7 \\
7 & 85 & 89 & 99.0 \\
8 & 59 & 76 & 99.3 \\
9 & 35 & 35 & 99.4 \\
10 & 44 & 53 & 99.6 \\
11 & 24 & 26 & 99.7 \\
12 & 20 & 20 & 99.7 \\
13 & 13 & 13 & 99.8 \\
14 & 9 & 19 & 99.9 \\
15 & 8 & 16 & 99.9 \\
16 & 10 & 10 & 99.9 \\
17 & 2 & 2 & 99.9 \\
18 & 3 & 3 & 100.0 \\
19 & 5 & 5 & 100.0 \\
20 & 1 & 1 & 100.0 \\
21 & 4 & 4 & 100.0 \\
24 & 1 & 1 & 100.0 \\
25 & 1 & 1 & 100.0 \\
32 & 1 & 1 & 100.0 \\
\midrule
all & 7{,}372 & 29{,}781 & \\
\bottomrule
\end{tabular}
\end{table}
 
\begin{table}[h]
\centering
\caption{All traded parlays on Kalshi on June 16, 2026 (12 hours), by leg count
($M=5{,}251$ base markets).}
\label{tab:flow_agg_app}
\begin{tabular}{rrrr}
\toprule
legs & markets & trades & cum.\ (\%) \\
\midrule
2     & 13{,}637 & 126{,}381 & 28.9 \\
3     & 24{,}088 &  92{,}702 & 50.0 \\
4     & 25{,}747 &  70{,}716 & 66.2 \\
5     & 19{,}930 &  35{,}695 & 74.3 \\
6     & 15{,}936 &  26{,}705 & 80.4 \\
7     & 11{,}235 &  17{,}202 & 84.3 \\
8     &  9{,}311 &  13{,}952 & 87.5 \\
9     &  6{,}454 &   9{,}121 & 89.6 \\
10    &  5{,}746 &   8{,}866 & 91.6 \\
11    &  4{,}137 &   9{,}422 & 93.8 \\
12    &  3{,}300 &   5{,}027 & 94.9 \\
13    &  2{,}439 &   3{,}743 & 95.8 \\
14    &  2{,}083 &   2{,}948 & 96.5 \\
15    &  2{,}122 &   6{,}610 & 98.0 \\
16    &  1{,}102 &   1{,}588 & 98.3 \\
17    &    957 &   1{,}324 & 98.6 \\
18    &    760 &   1{,}033 & 98.9 \\
19    &    546 &     700 & 99.0 \\
20    &    575 &     912 & 99.2 \\
21    &    422 &     573 & 99.4 \\
$22+$ &  1{,}905 &   2{,}795 & 100.0 \\
\midrule
all   & 152{,}432 & 438{,}015 & \\
\bottomrule
\end{tabular}
\end{table}
 
\begin{table}[h]
\centering
\caption{All traded parlays on Kalshi, September 16--23, 2026, by leg count
($M=191{,}806$ traded base markets). Orders above 50 are binned.}
\label{tab:flow_7d_app}
\scriptsize
\begin{tabular}{rrrr@{\hspace{2.5em}}rrrr}
\toprule
legs & markets & trades & cum.\ (\%) & legs & markets & trades & cum.\ (\%) \\
\midrule
2 & 433{,}523 & 1{,}823{,}930 & 16.0 & 28 & 14{,}689 & 20{,}304 & 97.6 \\
3 & 620{,}882 & 1{,}541{,}252 & 29.5 & 29 & 12{,}816 & 17{,}530 & 97.8 \\
4 & 666{,}595 & 1{,}334{,}216 & 41.1 & 30 & 12{,}819 & 18{,}278 & 97.9 \\
5 & 577{,}760 & 1{,}037{,}221 & 50.2 & 31 & 10{,}340 & 14{,}000 & 98.0 \\
6 & 561{,}479 & 979{,}270 & 58.8 & 32 & 9{,}208 & 18{,}739 & 98.2 \\
7 & 405{,}273 & 683{,}807 & 64.8 & 33 & 8{,}063 & 10{,}734 & 98.3 \\
8 & 375{,}084 & 613{,}300 & 70.1 & 34 & 7{,}166 & 10{,}336 & 98.4 \\
9 & 310{,}587 & 582{,}962 & 75.2 & 35 & 6{,}709 & 8{,}765 & 98.5 \\
10 & 294{,}773 & 521{,}054 & 79.8 & 36 & 6{,}105 & 9{,}027 & 98.5 \\
11 & 213{,}723 & 316{,}941 & 82.6 & 37 & 5{,}451 & 9{,}988 & 98.6 \\
12 & 190{,}051 & 274{,}370 & 85.0 & 38 & 5{,}033 & 7{,}262 & 98.7 \\
13 & 158{,}554 & 235{,}138 & 87.0 & 39 & 4{,}682 & 10{,}488 & 98.8 \\
14 & 140{,}932 & 236{,}263 & 89.1 & 40 & 4{,}915 & 7{,}315 & 98.8 \\
15 & 121{,}565 & 208{,}543 & 90.9 & 41 & 3{,}870 & 6{,}742 & 98.9 \\
16 & 92{,}807 & 142{,}685 & 92.2 & 42 & 3{,}628 & 6{,}623 & 99.0 \\
17 & 73{,}679 & 97{,}533 & 93.0 & 43 & 3{,}388 & 4{,}776 & 99.0 \\
18 & 63{,}293 & 87{,}703 & 93.8 & 44 & 3{,}074 & 4{,}203 & 99.0 \\
19 & 52{,}307 & 76{,}451 & 94.5 & 45 & 2{,}956 & 4{,}231 & 99.1 \\
20 & 52{,}748 & 77{,}124 & 95.1 & 46 & 2{,}632 & 5{,}362 & 99.1 \\
21 & 38{,}864 & 53{,}616 & 95.6 & 47 & 2{,}386 & 3{,}428 & 99.2 \\
22 & 33{,}888 & 53{,}992 & 96.1 & 48 & 2{,}323 & 4{,}077 & 99.2 \\
23 & 28{,}723 & 40{,}409 & 96.4 & 49 & 2{,}173 & 4{,}255 & 99.2 \\
24 & 25{,}448 & 33{,}565 & 96.7 & 50 & 2{,}420 & 3{,}633 & 99.3 \\
25 & 22{,}803 & 32{,}645 & 97.0 & 51--100 & 31{,}418 & 83{,}116 & 100.0 \\
26 & 18{,}846 & 24{,}713 & 97.2 & 101--200 & 831 & 1{,}239 & 100.0 \\
27 & 16{,}582 & 22{,}271 & 97.4 & 201--388 & 34 & 236 & 100.0 \\
\midrule
all & 5{,}759{,}898 & 11{,}425{,}661 & 100.0 & & & & \\
\bottomrule
\end{tabular}
\end{table}

\end{document}